\documentclass{IEEEtran}  
\usepackage{cite}
\usepackage{amsmath}
\usepackage{amssymb}
\usepackage{amsfonts}
\usepackage{algorithmic}
\usepackage{mathdots}
\usepackage{graphicx}
\usepackage{textcomp}
\usepackage{enumitem}
\usepackage{amsthm}
\usepackage{mathtools}
\usepackage[hidelinks]{hyperref} 
\usepackage{xurl}                

\theoremstyle{definition}
\newtheorem{Def}{Definition}
\theoremstyle{definition}
\newtheorem{The}{Theorem}
\theoremstyle{definition}
\newtheorem{Lem}{Lemma}
\theoremstyle{definition}

\theoremstyle{definition}
\newtheorem{Pro}{Proposition}
\theoremstyle{definition}
\newtheorem*{Problem}{Problem formulation}
\theoremstyle{remark}
\newtheorem{Rem}{Remark}
\theoremstyle{remark}
\newtheorem{Ex}{Example}
\theoremstyle{remark}

\theoremstyle{remark}
\newtheorem{Con}{Conjecture}
\theoremstyle{remark}

\theoremstyle{remark}

\usepackage{mdframed}

\DeclareFontEncoding{LS1}{}{}
\DeclareFontSubstitution{LS1}{stix2}{m}{n}
\DeclareMathAlphabet{\mathscr}{LS1}{stix2scr}{m}{n}
\newcommand{\err}{\ensuremath{ \mathscr{r} }}
\newcommand{\T}{\mathsf{T}}
\renewcommand{\i}{\mathbf{i}}
\renewcommand{\j}{\mathbf{j}}
\renewcommand{\k}{\mathbf{k}}

\renewcommand{\Re}{\operatorname{Re}}
\renewcommand{\Im}{\operatorname{Im}}

\newcommand{\RR}{\mathbb{R}}
\newcommand{\CC}{\mathbb{C}}
\newcommand{\HH}{\mathbb{H}}

\newcommand{\HHn}{\mathbb{H}^{n}}

\newcommand{\HHnn}{\mathbb{H}^{n\times n}}

\newcommand{\HHmn}{\mathbb{H}^{m\times n}}
\newcommand*{\conj}[1]{\bar{#1}}  
\newcommand*{\conjm}[1]{\overline{#1}} 

\def\BibTeX{{\rm B\kern-.05em{\sc i\kern-.025em b}\kern-.08em
    T\kern-.1667em\lower.7ex\hbox{E}\kern-.125emX}}

\newcommand{\narrowbmatrix}[1]{%
  {\setlength{\arraycolsep}{2pt}\begin{bmatrix}#1\end{bmatrix}}}
\newcommand{\DOI}[1]{\href{https://doi.org/#1}{DOI:\,\nolinkurl{#1}}}

\begin{document}


\title{Quaternionic Pole Placement via Companion Forms and the Ackermann Formula}
\author{Michael Sebek, \IEEEmembership{Life Senior Member, IEEE}
\thanks{Received September 2025;
This work was co-funded by the European Union under the project ROBOPROX (reg. no. 
CZ.02.01.01/00/22\_008/0004590);

The author thanks the Associate Editor and the anonymous reviewers for their careful reading and constructive comments.}
\thanks{M. Sebek is with the Czech Technical University in Prague, Faculty of Electrical Engineering, Department of Control Engineering, Prague, 16000, Czech Republic (e-mail: michael.sebek@fel.cvut.cz). }}
\maketitle
\begin{center}
\fbox{\begin{minipage}{0.94\linewidth}
\footnotesize
\textbf{Accepted manuscript / arXiv version.}
This is the accepted manuscript / author final version of: M. Sebek, "Quaternionic Pole Placement via Companion Forms and the Ackermann Formula," IEEE Transactions on Automatic Control, Early Access, DOI: 10.1109/TAC.2026.3702917. The article has been accepted for publication in a future issue of the journal and is available in IEEE Xplore. This arXiv version is not the IEEE copy-edited version of record; editorial changes may still be introduced before final issue assignment. The published IEEE article is expected to be available under a Creative Commons Attribution (CC BY) Open Access license.
The arXiv record is
\href{https://arxiv.org/abs/2509.21425}{arXiv:2509.21425}.
\end{minipage}}
\end{center}
\vspace{0.4em}
\begin{abstract}
We present an extension of state-feedback pole placement for quaternionic systems, based on companion forms and the Ackermann formula. For controllable single-input quaternionic LTI models, we define a companion polynomial that annihilates its companion matrix, characterize spectra via right-eigenvalue similarity classes, and prove coefficient-matching design in controllable coordinates. We then derive a coordinate-free Ackermann gain expression valid for real target polynomials, and state its scope and limitations. Short examples demonstrate correctness, practical use, and numerical simplicity.
\end{abstract}
\section{Introduction}
\label{Sec:Introduction}
Quaternions provide a natural extension of real and complex numbers: like $\mathbb{R}$ and $\mathbb{C}$, $\mathbb{H}$ is a division algebra, but—crucially—multiplication is \emph{noncommutative}. This changes some rules in systems and control. Eigenvalues must be interpreted carefully (right vs.\ left) and—on the right—they are defined only up to \emph{quaternionic similarity classes}. Determinants and a full Cayley–Hamilton theorem do not carry over in the usual way, and even polynomial evaluation is order-dependent, so \emph{right and left zeros need not coincide}. As a result, several classical tools are not yet fully developed—or are simply missing—when one works directly over $\mathbb{H}$; see, e.g., foundational results on quaternionic LTI stability \cite{PereiraVettori2006} and linear algebra \cite{Rodman2014}.

Quaternion models are a natural choice whenever states or operators are genuinely hypercomplex. Beyond rigid-body attitude kinematics, representative areas include:
spin-\(\tfrac{1}{2}\) 
quantum dynamics~\cite{WhartonKoch15} and qubit gate control—e.g., flatness-based single-qubit gates~\cite{daSilvaRouchon2008};
NMR/MRI pulse design—e.g., quaternion descriptions of selective-pulse propagators~\cite{Emsley1992};
polarization/bivariate signal processing and colour imaging—e.g., Quaternion-MUSIC for DOA–polarization~\cite{Miron2006}, quaternion Fourier transforms for colour images~\cite{Sangwine1996}, and a recent overview~\cite{SPMag2023};
and learning and computation with hypercomplex variables—e.g., quaternion neural controllers~\cite{Takahashi2017}, quaternion CNNs for speech~\cite{Parcollet2018}, and foundational links to quantum computing~\cite{Meglicki2008}.

Other interpretations of quaternions are already used in control of special classes of systems—for example, unit quaternions as rotation parameters, complex $2{\times}2$ representations, or real four-vector embeddings. In this note, we adopt the \emph{hypercomplex algebraic} viewpoint, where signals, states, and operators naturally reside in $\mathbb{H}$. This keeps the calculus aligned with the modeling (no repeated conversions to $\mathbb{R}$ or $\mathbb{C}$), makes spectral notions explicit (right eigenvalues and their similarity classes), and meshes naturally with companion constructions used in design.

The primary contribution of this work is to lift standard real/complex state-space methods to quaternionic systems. We tackle a basic, practical task: \emph{state-feedback pole placement} for single-input, controllable LTI models over $\HH$. The results rest on two pillars. First, we develop quaternionic \emph{companion forms} together with a quaternionic \emph{companion polynomial} that right-annihilates its companion matrix—playing the role of characteristic/minimal polynomials over fields—and thereby provide a clean, determinant-free framework in controllable (companion) coordinates. Second, we give two equivalent design routes: (i) \emph{coefficient matching} in the controllable companion form, and (ii) a \emph{coordinate-free Ackermann-type} expression based only on the controllability matrix and a desired companion polynomial. We also specify the exact scope: the Ackermann route applies when the target polynomial has \emph{real (central) coefficients}, which suffices to prescribe any combination of real classes and complex-conjugate classes of right eigenvalues.

These workflows mirror their $\RR/\CC$ counterparts---companion coordinates, coefficient matching, and an Ackermann-type gain---yet key quaternionic differences remain: right poles are \emph{similarity classes}, polynomial evaluation is order-sensitive, and the Ackermann step requires \emph{real (central)} target coefficients. We retain the familiar design flow while making these nuances explicit.

The note is organized as follows. Sections~II–IV develop the quaternionic preliminaries, LTI model, and companion-form framework. Sections~V–VII treat pole placement, coefficient matching, and the Ackermann formula, respectively; Section~VIII concludes.
\section{Quaternions and Matrices}
\paragraph*{Quaternions} A quaternion is a \emph{hypercomplex} scalar $q=a+b\i+c\j+d\k$
 with $a,b,c,d\in\RR$, and imaginary units $\i,\j,\k$ satisfying 
$\i^{2}=\j^{2}=\k^{2}=\i\j\k=-1$ (hence $\i\j=\k=-\j\i,$ $\j\k=\i=-\k\j,$ $\k\i=\j=-\i\k$). The set of all quaternions, denoted $\HH,$
is a 4-dimensional associative but \emph{noncommutative} division algebra over $\RR.$ For $q\in\HH,$ denote its real and imaginary parts by $\operatorname{Re}(q)=a$ and $\operatorname{Im}(q)=b\i+c\j+d\k$, respectively; the conjugate is $\conj{q}=a-b\i-c\j-d\k$ and the norm is $|q|=\sqrt{q\bar q}.$ Throughout we use the scalar (hypercomplex) model of \(\mathbb{H}\); see \cite{Rodman2014} for a more detailed introduction.

\paragraph*{Similarity and classes}
Two quaternions $q,r$ are \emph{similar}, written $q\sim r,$ if $q=\alpha^{-1} r\alpha$  for some nonzero $\alpha\in\HH.$
Similarity preserves $\operatorname{Re}(\cdot),$  $\left|\cdot \right|,$ and $\left|\operatorname{Im}(\cdot)\right| $ and partitions $\HH$ into  equivalence classes $\left[\cdot\right]$: real quaternions are isolated classes $\left[a\right]=\{a\}$; if $\operatorname{Im}(q)\neq0,$ then $[q]$ 
is a 2-sphere $\{\mathrm{Re}(q)+u\,\left|\operatorname{Im}(q)\right|:u\in\mathrm{Im}(\HH)\cong\RR^3,\ |u|=1\}$ that meets the standard complex line $\{a+b\i:\,a,b\in\mathbb{R}\}$ 
in exactly two points $\operatorname{Re}(q)\pm i\,\left|\operatorname{Im}(q)\right|.$
For example, $[\k]$ contains all imaginary units, including $\pm \i.$

\paragraph*{Quaternionic matrices}
For $m,n\in\mathbb{N}$, let $\HHmn$ denote the set of 
$m\times n$ matrices with entries in $\HH.$ We view column vectors as a right $\HH$-module and use the standard conjugate transpose $A^{*}=\overline{A}^{\T}.$ 
Noncommutativity shows up strongly: entrywise conjugation and transpose \emph{do not} distribute over products in the usual way: in general
$\conjm{AB} \neq \conj{A}\conj{B}$
and $(AB)^\T \neq B^\T A^\T$. By contrast, the mixed identity holds: $(AB)^\ast = B^\ast A^\ast.$ 
Likewise, \((A^{\top})^{-1}=(A^{-1})^{\top}\) \emph{need not} hold over \(\mathbb{H}\) (as transpose may destroy invertibility), whereas \((A^{*})^{-1}=(A^{-1})^{*}\) does hold. See e.g. \cite{Rodman2014} for further subtleties.

\paragraph*{Rank and invertibility}
For $A\in\HHmn$, the quaternionic rank, denoted $\operatorname{rank}_{\mathbb H}A$,
is the maximal number of \emph{right–independent} columns of $A$
(i.e., if $\sum_j a_j c_j = 0$ with coefficients $c_j\in\mathbb H$ on the \emph{right},
then all $c_j=0$). Equivalently, it is the maximal number of \emph{left–independent} rows. A square matrix
$A\in\mathbb H^{n\times n}$ is invertible iff $\operatorname{rank}_{\mathbb H}A=n,$ see \cite{Rodman2014}.

\paragraph*{Right eigenvalues and spectrum}
For $A\in\HHnn$, a \emph{right eigenpair} is $(v,\lambda)$
with $v\neq 0$ and 
\begin{IEEEeqnarray}{rCl}
    A v &=& v\,\lambda,\qquad v\in\mathbb{H}^{n},\ \lambda\in\mathbb{H}.
    \label{Eq0reigenpair}
\end{IEEEeqnarray}
Scaling $v$ on the right by $\alpha \neq 0 \in\HH$
replaces $\lambda$ by $\alpha^{-1}\lambda\alpha.$
So, right eigenvalues are defined up to quaternionic similarity.
Therefore, the right spectrum is the set of similarity classes
\[
\sigma_{\err}(A) := \{[\lambda] : \exists v\neq 0 \text{ such that } A v = v \lambda\} \subset \mathbb{H}/\sim,
\]
which is invariant under matrix similarity: \( A \mapsto S^{-1} A S \) \cite{Rodman2014}.
 When a scalar is needed, we choose the unique complex representative in each class with a nonnegative imaginary part and call it the \emph{standard eigenvalue.}
(Left eigenvalues, defined by $Av=\mu v$, are \emph{not} preserved by matrix similarity and will not be used unless explicitly stated.)
\paragraph*{Polynomials and evaluations}
Let $\HH[\lambda]$ denote the ring of quaternionic polynomials
$p(\lambda)=\sum_{k=0}^{n} p_k \lambda^{k}$ with $p_k\in\HH$, where the indeterminate $\lambda$
commutes with all coefficients.
Evaluation is order--dependent: for $q\in\HH$ and $A\in\HH^{m\times m}$ we define
\[
\begin{aligned}
p(q)_{\err}\!&:=\!\sum_{k=0}^{n}p_k q^{k},\quad
p(q)_{\ell}\!:=\!\sum_{k=0}^{n} q^{k}p_k,\\[-0.2ex]
p(A)_{\err}\!&:=\!\sum_{k=0}^{n} p_k A^{k}.
\end{aligned}
\]
i.e., in the right matrix evaluation, the matrix powers appear on the right of the coefficients.
Analogously one may set $p(A)_{\ell}:=\sum_{k=0}^{n} A^{k}p_k$ if needed.
If $p\in\RR[\lambda]$ (central coefficients), then $p(\cdot)_{\err}=p(\cdot)_{\ell}$ and we write simply $p(\cdot)$;
otherwise we keep the evaluation index explicit.
\paragraph*{Left and right zeros}
A \emph{right zero} of $p$ is $q\in\HH$ such that $p(q)_\err=0;$
a \emph{left zero} satisfies $p(q)_\ell=0$. 
Left and right zeros differ only by similarity.
If $p(q)_\ell=0$, then there is some $\alpha\neq0$
 with $p(\alpha^{-1} q\alpha)_\err=0$, and conversely \cite{GordonMotzkin65}.
 If all $p_k\in\mathbb{R}$ (central coefficients), then 
$p(q)_\err=p(q)_\ell$ for all $q$ and the left zeros equal the right zeros.

\paragraph*{Isolated and spherical right zeros}
Let \(p\in\HH[\lambda]\) be nonzero with \(n=\deg p\). A nonreal right zero is either \emph{isolated}
(a single point of its similarity class) or \emph{spherical} (the whole class \([q]\), hence infinitely many zeros).
The spherical case occurs exactly when the real quadratic
\(\chi_q(\lambda)=\lambda^2-2\Re(q)\lambda+|q|^2\in\RR[\lambda]\) divides \(p(\lambda)\);
thus \(p\) has infinitely many right zeros iff it has a spherical right zero.
If \(p\) has no spherical right zeros, then all right zeros are isolated and their (algebraic) multiplicities sum to \(n=\deg p\).
In general, the same total \(n\) is obtained by counting isolated zeros with (algebraic) multiplicity and counting each spherical class twice (with spherical multiplicity). \cite{PogoruiShapiro2004,Opfer09}
\paragraph*{Notation}
We view \(\HH^{n}\) as a right \(\HH\)-module. Let \(e_i\in\RR^{n}\subset\HH^{n}\) be the \(i\)th canonical basis vector (so \(e_n=[0,\ldots,0,1]^{\top}\)). Since \(e_i\) has real entries, left/right scalar multiplication is unambiguous.
\section{Quaternionic state–space systems}
We study continuous-time LTI state–space models over $\mathbb H$
\begin{IEEEeqnarray*}{rCl}
\dot{x}(t) &=& Ax(t)+Bu(t), \IEEEyesnumber\label{Eq0AB}
\end{IEEEeqnarray*}
and the output equation plays no role in what follows. For clarity, we restrict to the single-input case: $A\in\mathbb H^{n\times n}$, $B\in\mathbb H^{n}$, with state $x(t)\in\mathbb H^{n}$ and input $u(t)\in\mathbb H$. (When desired, an output $y(t)=Cx(t)+Du(t)$ with $C\in\mathbb H^{p\times n}$ and $D\in\HH^{p}$ can be appended without affecting the development.)
\paragraph*{Stability}
Asymptotic (internal) stability is defined exactly as in the real/complex case: with \(u\equiv 0\), every trajectory satisfies \(x(t)\to 0\) for all \(x(0)\in\HH^{n}\).
Equivalently, \eqref{Eq0AB} is asymptotically stable iff every right–eigenvalue class of \(A\) meets \(\CC\) in a point \(\lambda\) with \(\operatorname{Re}\,\lambda<0\); see \cite{PereiraVettori2006}.

\paragraph*{Controllability}
As in the real/complex case, controllability means the ability to steer the state from any \(x_0\in\HHn\) to any \(x_f\in\HHn\) in finite time by a suitable input.
In the single–input setting, the controllability matrix is
\begin{IEEEeqnarray}{rCl}
\mathcal C &=& \bigl[\, B \;\; AB \;\; \cdots \;\; A^{n-1}B \,\bigr] \in \mathbb H^{n\times n},
\label{Eq0CONmat}
\end{IEEEeqnarray}
and the pair $(A,B)$ is completely controllable iff $\mathcal C$ is invertible over $\mathbb H$
(equivalently, iff $\operatorname{rank}_{\mathbb H}\mathcal C = n$) \cite{JiangLiuKouWang20}.

\section{Companion Form and Companion Polynomial}
In classical (real/complex) linear systems, two key tools are the controllable companion form and the characteristic or minimal polynomial. For quaternionic systems, the controllable companion form still applies. However, characteristic and minimal polynomials do not carry over directly, because there is no simple determinant on quaternions. Instead, we can work with a quaternionic companion polynomial.

\begin{Def}[\emph{Controllable companion form and companion polynomial}]
Let $A_c\in\HH^{n\times n}$ and $B_c\in\HH^{n}$ be
\begin{IEEEeqnarray}{rClrCl}
A_c &=& \begin{bmatrix}
   0&1&\ldots&0\\
        \vdots&\vdots&\ddots&\vdots\\
         0&0&\ldots&1\\
        -a_0&-a_1&\ldots&-a_{n-1}\\       
    \end{bmatrix},\;\;
&B_c &=& \begin{bmatrix}
0 \\
\vdots \\
0 \\
1
\end{bmatrix}.
\label{Eq0AcBcI}
\end{IEEEeqnarray}
with $a_0,\ldots,a_{n-1}\in\HH$. We say that $(A_c,B_c)$ is in (lower) \emph{ controllable companion (canonical) form.} The \emph{companion polynomial} associated with $(A_c,B_c)$ (or with $A_c$) is the monic quaternionic polynomial
\begin{IEEEeqnarray}{rCl}
a(\lambda)&=&a_{0}+a_{1}\lambda+\cdots+a_{n-1}\lambda^{n-1}+\lambda^{n}\in\HH[\lambda].
\end{IEEEeqnarray}
\end{Def}
If \(A\in\HHnn\) is \emph{cyclic}—i.e., there exists \(B\in\HHn\) with 
\(\mathcal{C}\) invertible—then \(A\) is similar to the lower controllable companion form. The next theorem makes this precise and gives the construction.
\begin{The}[Determinant-free companion form for a controllable pair]\label{ThCompanion}
Let \(A \in \HH^{n\times n}\) and \(B \in \HH^{n}\) be given, and assume the
controllability matrix \(\mathcal{C}=[B,AB,\dots,A^{n-1}B]\) is invertible over \(\HH\).

\begin{enumerate}
\item Then there exists an invertible similarity \(T\in\HH^{n\times n}\) such that
\begin{equation}\label{Eq0sim}
A_c=T^{-1}AT,\qquad B_c=T^{-1}B,
\end{equation}
and the pair \((A_c,B_c)\) is in lower controllable companion form \eqref{Eq0AcBcI}.
In particular, the coefficients \(a_0,\dots,a_{n-1}\in\HH\) are defined by the last row of \(A_c\),
\[
e_n^\top A_c=\begin{bmatrix}-a_0&-a_1&\cdots&-a_{n-1}\end{bmatrix}.
\]

\item The matrix \(T\) is explicitly given by
\begin{IEEEeqnarray}{rCl}
T^{-1} &=&
\begin{bmatrix}
t\\ tA\\ \vdots\\ tA^{n-1}
\end{bmatrix},
\label{Eq0Tidef}
\end{IEEEeqnarray}
where the first row reads
\begin{IEEEeqnarray}{rCl}
t &=& [0,0,\ldots,1]\,\mathcal{C}^{-1}.
\label{Eq0trowdef}
\end{IEEEeqnarray}

\item For the given pair \((A,B)\), the companion pair \((A_c,B_c)\) and the similarity \(T\) are unique.
\end{enumerate}
\end{The}
\begin{Rem}
Compared with the classical \(\mathbb{R} /\mathbb{C}\) case, the construction is \emph{completely reversed}: rather than deriving the companion realization from a characteristic/minimal polynomial computed via determinants (as usual in the classical \(\mathbb{R}/\mathbb{C}\) setting), we first construct the controllable companion form directly from the controllability matrix and only then \emph{read off} the companion polynomial from its last row. This route is determinant–free and sidesteps characteristic/minimal polynomials, Cayley–Hamilton, and adjugates, which do not transfer cleanly to \(\mathbb{H}\).
\end{Rem}
\begin{proof}
The proof follows the standard companion-form argument \cite{AntsaklisMichel06}, adapted to \(\HH\).

Since \(\mathcal C\) is invertible, \(t=e_n^\top \mathcal C^{-1}\) is well defined, and so is \(T^{-1}\) from \eqref{Eq0Tidef}. For \(i,j=1,\dots,n\),
\[
(T^{-1}\mathcal C)_{ij}=tA^{i+j-2}B.
\]
Because \(t\) is the last row of \(\mathcal C^{-1}\), we have \(tA^kB=0\) for \(k=0,\dots,n-2\) and \(tA^{n-1}B=1\). Hence \(T^{-1}\mathcal C\) is unit anti-triangular, so \(T\) is invertible.

Set \(M:=T^{-1}AT\). First, since the first column of \(\mathcal C\) is \(B\),
\[
T^{-1}B=\begin{bmatrix} tB & tAB & \cdots & tA^{n-1}B \end{bmatrix}^{\!\top}=e_n.
\]

Next, let \(W=[A^{n-1}B,\dots,B]\) and \(Q:=T^{-1}W\). Since both \(T\) and \(W\) are invertible, so is \(Q\). For \(i=1,\dots,n-1\) and \(j=1,\dots,n\),
\[
e_i^\top MQe_j=e_i^\top T^{-1}AWe_j=tA^{n+i-j}B,
\]
while
\[
e_{i+1}^\top Qe_j=e_{i+1}^\top T^{-1}We_j=tA^{n+i-j}B.
\]
Thus \((e_i^\top M-e_{i+1}^\top)Qe_j=0\) for all \(j\), and invertibility of \(Q\) gives \(e_i^\top M=e_{i+1}^\top\) for \(i=1,\dots,n-1\). Therefore the first \(n-1\) rows of \(M\) are exactly the shift rows of the lower controllable companion form.

Writing the last row as \(e_n^\top M=[-a_0,\,-a_1,\,\dots,\,-a_{n-1}]\), we conclude that \(M\) has the form \eqref{Eq0AcBcI}. Hence \(A_c=T^{-1}AT\) and \(B_c=T^{-1}B=e_n\).

For uniqueness, let \(\widetilde T\) be any invertible matrix such that \(\widetilde T^{-1}B=e_n\) and \(\widetilde M:=\widetilde T^{-1}A\widetilde T\) is in lower controllable companion form. Put \(\widetilde t:=e_1^\top \widetilde T^{-1}\). Since \(e_i^\top \widetilde M=e_{i+1}^\top\) for \(i=1,\dots,n-1\), we get \(e_i^\top \widetilde T^{-1}=\widetilde tA^{i-1}\), so \(\widetilde T^{-1}\) has the same form as \eqref{Eq0Tidef}, with \(\widetilde t\) in place of \(t\).

Now
\[
\widetilde t\,\mathcal C
= e_1^\top \widetilde T^{-1}\mathcal C
= \begin{bmatrix}
e_1^\top e_n &
e_1^\top \widetilde M e_n &
\cdots &
e_1^\top \widetilde M^{n-1} e_n
\end{bmatrix}.
\]
Because \(\widetilde M\) has the companion shift rows, \(e_1^\top \widetilde M^k=e_{k+1}^\top\) for \(k=0,\dots,n-1\), hence \(\widetilde t\,\mathcal C=e_n^\top\). Therefore \(\widetilde t=e_n^\top \mathcal C^{-1}=t\), so \(\widetilde T^{-1}=T^{-1}\), \(\widetilde T=T\), and the pair \((A_c,B_c)\) is unique.
\end{proof}
\begin{Ex}\label{ExCompanion}
Let
\begin{IEEEeqnarray}{rCl}\label{Eq0ABEx}
A&=&\begin{bmatrix} 1 & \i\\[2pt] \j & \k \end{bmatrix},\quad
B=\begin{bmatrix} 1\\[2pt] \k \end{bmatrix}.
\end{IEEEeqnarray}
Then the controllability matrix and its inverse are
\begin{IEEEeqnarray}{rCll}
    \mathcal{C}&=&&
    \begin{bmatrix}
    1 & 1-\j \\
     \k & -1+\j
    \end{bmatrix} ,\label{Eq0CON}\\
    \mathcal{C}^{-1} &=& \tfrac14&
    \begin{bmatrix}
    2-2\k  &  2-2\k \\             
    1+\i+\j+\k & -1+\i-\j+\k
    \end{bmatrix}. \label{Eq0CONi} \IEEEyesnumber
\end{IEEEeqnarray}
%
Hence
\begin{IEEEeqnarray*}{rCl}
    t = \tfrac14
    \begin{bmatrix}    
    1+\i+\j+\k & -1+\i-\j+\k
    \end{bmatrix}
\end{IEEEeqnarray*}
and
\begin{IEEEeqnarray}{rCll}
    T^{-1}&=& \tfrac14&
    \begin{bmatrix}
      1+\i+\j+\k & -1+\i-\j+\k\\
     2+2\k & -2-2\k
    \end{bmatrix},
    \label{Eq0ExTi}\\
     T&=&&
    \begin{bmatrix}
     -\i-\k  &   1\\
    -\i-\k   &  \k
    \end{bmatrix}. \label{Eq0ExT}
\end{IEEEeqnarray}
The lower controllable companion form is
{\setlength{\arraycolsep}{2pt}
\begin{IEEEeqnarray}{rCl}\label{Eq0ExAcBc}
A_c &=& 
\begin{bmatrix}
     0    &   1 \\   
       1-\i+\j-\k   &  1+\i-\j+\k
\end{bmatrix}, \,
B_c=\begin{bmatrix} 0\\[2pt] 1 \end{bmatrix},   
\end{IEEEeqnarray}
}
and the associated companion polynomial is
\begin{IEEEeqnarray}{rCl}\label{Eq0Excomppol}
a(\lambda) &=& -1+\i-\j+\k -  (1+\i-\j+\k)\lambda +\lambda^2.
\end{IEEEeqnarray}
\end{Ex}
We can still have a clean annihilating identity for the lower controllable companion even though full Cayley–Hamilton fails over $\HH.$ See \cite{ChapmanMachen16}.

\begin{The}[\emph{Companion polynomial annihilates its companion matrix over $\HH$}]\label{The0Companionannihilates}
Let $A_c\in\HHnn$ be in the companion form \eqref{Eq0AcBcI}.
Then the companion polynomial
\begin{IEEEeqnarray*}{rCl}  
  a(\lambda) &=& a_0 +a_{1}\lambda+\cdots+a_{n-1}\lambda^{n-1}+ \lambda^n \in\HH[\lambda]
\end{IEEEeqnarray*}
satisfies the right matrix evaluation identity
\begin{IEEEeqnarray*}{rCl}
a_0 I + a_1 A_c + \cdots + a_{n-1}A_c^{n-1} + A_c^n=0 .
\end{IEEEeqnarray*}
\end{The}
\begin{proof}
Define the row sequence
\begin{IEEEeqnarray*}{rCl}
r_k&:=&e_1^\top A_c^k,\qquad k\ge 0.
\end{IEEEeqnarray*}
Then
\begin{IEEEeqnarray*}{rCl}
r_0=e_1^\top,\; r_1=e_2^\top,\; \ldots,\; r_{n-1}=e_n^\top,
\end{IEEEeqnarray*}
so \(r_0,\dots,r_{n-1}\) are exactly the rows of the identity matrix and in particular form a basis of \(\HH^{1\times n}\).

By the shape of the companion matrix,
\begin{IEEEeqnarray*}{rCl}
r_n&=&e_n^\top A_c=\big[-a_0,\dots,-a_{n-1}\big]
    =-\sum_{k=0}^{n-1} a_k\,r_k.
\end{IEEEeqnarray*}
Since \(r_{k+1}=r_kA_c\), it follows that for every \(j\ge 0\),
\begin{IEEEeqnarray}{rCl}
r_{n+j}
&=&r_nA_c^j
=-\sum_{k=0}^{n-1} a_k\,r_kA_c^j
=-\sum_{k=0}^{n-1} a_k\,r_{k+j}.
\label{Eq0rowrec}
\end{IEEEeqnarray}

Now evaluate \(a(\lambda)\) at \(A_c\) using the right matrix evaluation:
\begin{IEEEeqnarray*}{rCl}
a(A_c)_{\err}&=&a_0I+a_1A_c+\cdots+a_{n-1}A_c^{n-1}+A_c^n.
\end{IEEEeqnarray*}
For each \(j=0,\dots,n-1\),
\begin{IEEEeqnarray*}{rCl}
r_j a(A_c)_{\err}
&=&a_0r_j+a_1r_{j+1}+\cdots+a_{n-1}r_{j+n-1}+r_{j+n}\\
&=&0
\end{IEEEeqnarray*}
by \eqref{Eq0rowrec}. Hence every row \(r_j\), \(j=0,\dots,n-1\), annihilates \(a(A_c)_{\err}\).
Since \(r_0,\dots,r_{n-1}\) are the rows of the identity matrix, this implies \(a(A_c)_{\err}=0\).
\end{proof}
\begin{Rem}
This \emph{is not} the (generally false) \emph{quaternionic Cayley–Hamilton theorem;} it is a one-sided polynomial identity tied to the companion’s shift structure.
\end{Rem}
The right spectrum of the companion matrix is the set of similarity classes determined by the right zeros of the companion polynomial.
\begin{The}
[\emph{Right zeros form the right spectrum of $A_c$}]\label{ThAcBcroots}
Let $(A_c,B_c)$ be in the quaternionic companion form \eqref{Eq0AcBcI} and
\begin{IEEEeqnarray*}{rCl}
a(\lambda) &=& a_0+a_1\lambda+\cdots+a_{n-1}\lambda^{n-1}+\lambda^n
\in\HH[\lambda]
\end{IEEEeqnarray*}
be the related quaternionic companion polynomial. Then
\begin{enumerate}[leftmargin=*,noitemsep]
\item If $\lambda\in\HH$ is a \emph{right zero} of $a(\lambda),$ then $\lambda$ is also a \emph{right eigenvalue} of $A_c$ and $v=[1,\lambda,\ldots,\lambda^{n-1}]^\T$ satisfies $A_c v=v\lambda$. 
\item Consequently,
\[
\sigma_\err(A_c)=\{[\lambda]:a(\lambda)_{\err}=0\}.
\]
Equivalently, if \(\lambda_1,\dots,\lambda_m\) are the distinct right zeros of \(a\) modulo similarity, then
\[
\sigma_\err(A_c)=\{[\lambda_1],\dots,[\lambda_m]\}.
\]
\item Every \emph{left zero} $\mu$ of $a(\lambda)$ (under left evaluation) is a \emph{right eigenvalue} of $A_c$ as well. 
\end{enumerate}
\end{The}
\begin{proof}
\noindent\emph{Step~1:} Put \(v=\left[1,\lambda,\ldots,\lambda^{n-1}\right]^\T\) and assume \(a(\lambda)_{\err}=0\). By construction,
\begin{IEEEeqnarray*}{rCl}
(A_c v)_1 &=& \lambda,\ \ldots,\ (A_c v)_{n-1}=\lambda^{n-1},\\
(A_c v)_n&=&-(a_0+a_1\lambda+\cdots+a_{n-1}\lambda^{n-1})=\lambda^n.
\end{IEEEeqnarray*}
Thus \(A_c v=v\lambda\).\\
\noindent\emph{Step 2:} By Step~1, every right zero \(\lambda_i\) yields a right eigenpair
\(A_c v=v\lambda_i\). Hence the class \([\lambda_i]\) belongs to the right spectrum.
Indeed, if \(A_cx=x\lambda_i\), then for every \(q\in\HH^\times\),
\begin{IEEEeqnarray*}{rCl}
A_c(xq)&=&(A_cx)q=x\lambda_i q=(xq)(q^{-1}\lambda_i q),
\end{IEEEeqnarray*}
so every element of \([\lambda_i]\) is a right eigenvalue of \(A_c\).

For the converse, take a right eigenpair \(A_cx=x\mu\), \(x\neq 0\), and write
\(x=[x_1,\dots,x_n]^\T\). From the first \(n-1\) rows,
\begin{IEEEeqnarray*}{rCl}
x_{k+1}&=&x_k\mu,\qquad k=1,\dots,n-1,
\end{IEEEeqnarray*}
hence
\begin{IEEEeqnarray*}{rCl}
x_k&=&x_1\mu^{k-1},\qquad k=1,\dots,n.
\end{IEEEeqnarray*}
Moreover, \(x_1\neq 0\), since otherwise the recursion would give \(x_2=\cdots=x_n=0\), contradicting \(x\neq 0\).

Using the last row of \(A_cx=x\mu\) and the above relations, we obtain
\begin{IEEEeqnarray*}{rCl}
-(a_0x_1+a_1x_1\mu+\cdots+a_{n-1}x_1\mu^{n-1})&=&x_1\mu^n.
\end{IEEEeqnarray*}
Right-multiplying by \(x_1^{-1}\) gives
\begin{IEEEeqnarray*}{rCl}
-a_0-a_1(x_1\mu x_1^{-1})-\cdots-a_{n-1}(x_1\mu^{n-1}x_1^{-1})&=&x_1\mu^n x_1^{-1}.
\end{IEEEeqnarray*}
Let \(\lambda:=x_1\mu x_1^{-1}\). Then \(x_1\mu^k x_1^{-1}=\lambda^k\) for all \(k\), so the equality becomes
\begin{IEEEeqnarray*}{C}
-a_0-a_1\lambda-\cdots-a_{n-1}\lambda^{n-1}=\lambda^n
\;\iff\;
a(\lambda)_{\err}=0.
\end{IEEEeqnarray*}
Thus every right eigenvalue \(\mu\) is similar to a right zero \(\lambda\) of \(a(\lambda)\), so the class \([\mu]\) equals \([\lambda]\). Therefore every element of \(\sigma_\err(A_c)\) is of the form \([\lambda]\) with \(a(\lambda)_{\err}=0\). Together with the first inclusion, this yields
\begin{IEEEeqnarray*}{rCl}
\sigma_\err(A_c)&=&\{[\lambda]:a(\lambda)_{\err}=0\}.
\end{IEEEeqnarray*}
\noindent\emph{Step 3:} The left zeros of \(a(\lambda)\) are known to be similar to its right zeros \cite{GordonMotzkin65}. By Step~2, the right spectrum of \(A_c\) contains the entire similarity class of every right zero. Hence every left zero of \(a(\lambda)\) is a right eigenvalue of \(A_c\) as well.
\end{proof}

The companion polynomial can actually be associated with any cyclic quaternionic matrix $A$ using its controllable companion form $A_c.$

\begin{Def}[\emph{Companion polynomial of a controllable pair}]
Let $(A,B)$ with $A\in\HH^{n\times n}$ and $B\in\HH^{n}$ satisfy that the
controllability matrix $\mathcal C=[B,AB,\dots,A^{n-1}B]$ is invertible over $\HH$.
Let $(A_c,B_c)$ be the lower controllable companion form of $(A,B)$; this form is unique. The \emph{companion polynomial}
of $(A,B)$ is the monic quaternionic polynomial
\begin{IEEEeqnarray}{rCl}
a(\lambda)&=&a_{0}+a_{1}\lambda+\cdots+a_{n-1}\lambda^{n-1}+\lambda^{n},
\label{Eq0comppolgeneral}
\end{IEEEeqnarray}
where $\left[-a_0,\dots,-a_{n-1}\right]$ is the last row of $A_c$.  
\end{Def}
\begin{Rem}
Throughout this note, we adopt the \emph{lower} controllable companion form.
Other orientations are related by permutation similarity, and therefore yield the same right spectrum. In this paper we fix the lower orientation, which uniquely determines the associated companion polynomial for a controllable pair.
\end{Rem}
\begin{The}[\emph{Right zeros form the right spectrum of $A$}] \label{ThABroots}
Let $(A,B)$ be a controllable quaternionic pair and $a(\lambda)$ 
be its quaternionic companion polynomial. Then
\begin{enumerate}[leftmargin=*,noitemsep]
\item If $\lambda\in\HH$ is a right zero of $a(\lambda),$
then $\lambda$ is also a right eigenvalue of $A.$ 
\item Consequently,
\[
\sigma_\err(A)=\{[\lambda]:a(\lambda)_{\err}=0\}.
\]
\item Every left zero $\mu$ of $a$ (under left evaluation) is a right eigenvalue of $A$.
\end{enumerate}
\end{The}
\begin{proof} Since \(A\) is similar to \(A_c\), they share the right spectrum; the claims follow  from Theorem \ref{ThAcBcroots}.
\end{proof}
\begin{Rem}
Hence, the companion polynomial shares some properties of the minimal/characteristic polynomial even in quaternions, but not all! In particular, it \emph{does not} annihilate $A$ in general; it only annihilates its companion form $A_c$. Hence \(a(A_c)_{\err}=0\) but, in general, \(a(A)_{\err}\neq 0\).
\end{Rem}
\begin{Rem}
For a fixed controllable pair \((A,B)\), the associated lower companion realization is unique, and therefore so is its monic companion polynomial. When \(A\) is cyclic, one may also speak informally of a companion polynomial of \(A\), but this polynomial is not intrinsic to \(A\) alone: different cyclic choices of \(B\) generally produce different companion polynomials. What is invariant is the set of similarity classes determined by their right zeros, namely the right spectrum \(\sigma_\err(A)\).
\end{Rem}

\section{Control Problem: State-Feedback Pole Placement in Quaternionic Systems}
We consider single-input LTI quaternionic state–space models \eqref{Eq0AB} with
\begin{IEEEeqnarray}{rCl}
\dot{x}(t)&=&A x(t)+B u(t), 
\label{eq:ss-model}
\end{IEEEeqnarray}
with $A\in\HHnn$ and $B\in\HHn$.  
We use static state feedback
\begin{IEEEeqnarray*}{rCl}
u(t)&=&-K x(t), \qquad K\in\mathbb{H}^{1\times n},
\end{IEEEeqnarray*}
which yields the closed-loop dynamics
\begin{IEEEeqnarray*}{rCl}
\dot{x}(t)=A_{\mathrm{cl}} x(t), 
\qquad A_{\mathrm{cl}}\coloneqq A-BK.
\end{IEEEeqnarray*}

In the real or complex case, the method assigns $n$ (not necessarily distinct) eigenvalues of $A_{\mathrm{cl}}$. In the quaternionic setting, the situation is different: spectra must be interpreted through \emph{right eigenvalues} \eqref{Eq0reigenpair} which are determined only up to quaternionic similarity (i.e., $\lambda \sim \alpha^{-1}\lambda\,\alpha$ for any $\alpha\in\mathbb{H}\setminus\{0\}$). 
Consequently, pole placement specifies the \emph{closed-loop right spectrum} as the set of such similarity classes. Equivalently, one may specify canonical complex representatives
\begin{IEEEeqnarray*}{c}
\lambda_k\in\CC,\qquad \Im \lambda_k\ge 0,\qquad k=1,\dots,n,
\end{IEEEeqnarray*}
counted with algebraic multiplicity.

\begin{Problem}[\emph{Pole placement over $\HH$}]
Given $A\in\mathbb H^{n\times n}$, $B\in\mathbb H^{n}$, and desired representatives
\(\lambda_1,\dots,\lambda_n\in\mathbb H\) (not necessarily distinct), or equivalently their canonical complex representatives,
find $K\in\mathbb H^{1\times n}$ such that, for the closed loop $A_{\rm cl}=A-BK$,
\begin{IEEEeqnarray*}{rCl}
\sigma_\err(A_{\mathrm{cl}})&=&\{[\mu_1],\dots,[\mu_m]\},
\end{IEEEeqnarray*}
where \(\mu_1,\dots,\mu_m\) are representatives of the distinct desired similarity classes appearing among the targets. Any desired multiplicities are encoded separately by the chosen polynomial data.
\end{Problem}
Typically, one chooses a monic closed-loop companion polynomial \(a_d(\lambda)\) whose right zeros represent the desired similarity classes. Thus \(a_d\) is a convenient algebraic encoding of the pole-placement target, while the actual control objective is the closed-loop right spectrum. When \(a_d\in\mathbb{R}[\lambda]\), nonreal zeros appear in conjugate pairs and encode one spherical class each. For continuous-time systems, one usually picks canonical complex representatives in the open left half-plane; \(\operatorname{Re}\) sets the decay rate and \(|\operatorname{Im}|\) the oscillation frequency.
\section{Quaternionic Pole Placement via Controllable Companion Form}

In the real/complex setting, state feedback is obtained by transforming a controllable pair to companion form and matching coefficients. Over \(\mathbb{H}\), the same workflow applies for single-input controllable systems: use the \emph{companion polynomial}—read from the last row of the companion matrix and computed over \(\mathbb{H}\)—and interpret poles as \emph{right-eigenvalue similarity classes}. Pole placement then reduces to a coefficient update in controllable-companion coordinates, followed by mapping the gain back through the similarity transformation. Crucially, this construction avoids determinants and characteristic/minimal polynomials—as well as Cayley–Hamilton and adjugates—which do not carry over cleanly to \(\mathbb{H}\).
\begin{The}[\emph{Pole placement via controllable companion form}]
Assume $(A,B)$ with $A\in\HH^{n\times n}$, $B\in\HH^{n}$ is controllable, and let
$(A_c,B_c)$ be its (lower) controllable companion form via $T$ so that
$A_c=T^{-1}AT$, $B_c=T^{-1}B$.
Write the last row of $A_c$ as $\left[-a_0,\dots,-a_{n-1}\right]$ and let
$a(\lambda)=a_0+a_1\lambda+\cdots+a_{n-1}\lambda^{n-1}+\lambda^n$ be the companion polynomial.
Given any monic quaternionic polynomial
\begin{IEEEeqnarray}{rCl}
  a_d(\lambda)&=&d_{0}+d_{1}\lambda+\cdots+d_{n-1}\lambda^{n-1}+\lambda^{n}
  \in\HH[\lambda],
\end{IEEEeqnarray}
there exists a unique feedback $u=-Kx$, $K\in\HH^{1\times n}$, such that the closed-loop companion polynomial of \(A_{\mathrm{cl},c}:=A_c-B_cK_c\), with \(K_c:=KT\), equals \(a_d\). Consequently, the right spectrum of \(A_{\mathrm{cl}}:=A-BK\) is the set of similarity classes determined by the right zeros of \(a_d\). The multiplicities are encoded by the polynomial \(a_d\) itself, not by the set-valued spectrum. Moreover
\begin{IEEEeqnarray}{rCl}
K &=& K_c\,T^{-1},\\
K_c &=& \big[d_0-a_0\ \ d_1-a_1\ \ \cdots\ \ d_{n-1}-a_{n-1}\big].
\end{IEEEeqnarray}
\end{The}

\begin{proof}
In controllable coordinates, with $K_c:=KT$, we have
\[
A_{\mathrm{cl},c}:=T^{-1}(A-BK)T=A_c-B_cK_c.
\]
Writing $K_c=[k_0\ k_1\ \cdots\ k_{n-1}]$, the closed-loop companion polynomial $a_{\mathrm{cl}}(\lambda)$ reads
\begin{IEEEeqnarray*}{C}
(a_0+k_0)+(a_1+k_1)\lambda+\cdots+(a_{n-1}+k_{n-1})\lambda^{n-1}+\lambda^n.
\end{IEEEeqnarray*}
Setting $k_i=d_i-a_i$ for $i=0,\dots,n-1$ yields $a_{\mathrm{cl}}(\lambda)=a_d(\lambda)$.
By the previously established spectral property of the companion form, the right spectrum of
$A_{\mathrm{cl},c}$ is the set of similarity classes determined by the right zeros of
$a_{\mathrm{cl}}$, hence by those of $a_d$. Since $A_{\mathrm{cl}}$ is similar to
$A_{\mathrm{cl},c}$, the same holds for $A_{\mathrm{cl}}$.

Uniqueness holds for the prescribed polynomial $a_d$: the coefficient update
$k_i=d_i-a_i$ is bijective, and $K\mapsto K_c=KT$ is bijective because $T$ is invertible.
Hence there is a unique feedback whose closed-loop companion polynomial equals $a_d$.
Finally $K=K_cT^{-1}$.
\end{proof}

\begin{Ex}[\emph{Example \ref{ExCompanion} continued}]\label{ExSSCompan}
With the companion form \eqref{Eq0ExAcBc}, the similarity matrices \eqref{Eq0ExTi}, and the companion polynomial \eqref{Eq0Excomppol} in hand,  we now consider several choices of desired polynomials:\\
\emph{(a) Real targets:}  $a_d(\lambda)=\lambda^2+3\lambda+2$ ($\lambda_1=-1,\lambda_2=-2$) gives
\begin{IEEEeqnarray*}{lCl}
    K_c&=&
    \begin{bmatrix}
      3-\i+\j-\k & 4+\i-\j+\k
    \end{bmatrix}, \\
    K&=&
    \begin{bmatrix}
    2.5+\i+2.5\k  &  -1.5+\i-1.5\k
    \end{bmatrix},  \IEEEyesnumber \label{Eq0ExKreal}
\end{IEEEeqnarray*}
and the closed-loop matrix
\begin{IEEEeqnarray*}{rCl}
    A_{\rm cl} &=& 
    \begin{bmatrix}
    -1.5-\i-2.5\k & 1.5+1.5\k \\  
     2.5-2.5\k  & -1.5-\j +2.5\k
    \end{bmatrix}    \IEEEyesnumber \label{Eq0ExAclreal}
\end{IEEEeqnarray*}
has the right spectrum $\sigma_\err(A_{\rm cl}) = \{[-1],[-2]\}$ as desired.\\
\emph{(b) Complex pair:} 
 $a_d(\lambda)=\lambda^2+2\lambda+2$
 ($\lambda_{1,2}=-1\pm \i$) yields
\begin{IEEEeqnarray*}{lCl}
 K_c&=&
    \begin{bmatrix}
    3-\i+\j-\k  &   3+\i-\j+\k
    \end{bmatrix} ,    \\
 K&=&
    \begin{bmatrix}
    2+\i+2\k & -1+\i-\k  
    \end{bmatrix} ,   \IEEEyesnumber \label{Eq0ExKcomplex} 
\end{IEEEeqnarray*}
and
\begin{IEEEeqnarray}{rCl}\label{Eq0ExAclcomplex}
 A_{\rm cl} &=&
    \begin{bmatrix}
    -1-\i-2\k & 1+\k \\
     2-2\k & -1-\j+2\k
    \end{bmatrix} .   
\end{IEEEeqnarray}
Here \(\sigma_\err(A_{\mathrm{cl}})=\{[-1+\i]\}\); equivalently, \(-1\pm\i\) are the two canonical complex representatives of the assigned spherical class.\\
\emph{(c) Quaternionic coefficients:}
All displayed numerical values in this example are rounded independently from full-precision computations; therefore, the printed intermediate quantities are not intended to reproduce one another exactly.
Demanding quaternionic target classes, represented for example by \(\lambda_1=-1+\j\) and \(\lambda_2=-2+\k\), yields the following rounded target polynomial
\begin{IEEEeqnarray}{rCl}
    a_d(\lambda)&=&
    2.7-\i-1.3\j+0.33\k \IEEEnonumber \\
    && + (3-0.67\i-0.33\j-0.33\k)\lambda + \lambda^2.
    \label{Eq0Exadlambdaquat}
  \end{IEEEeqnarray}
Here \(a_d(\lambda)\in\mathbb{H}[\lambda]\) has nonreal (noncentral) coefficients, but the companion-form coefficient update still applies. Using
\(K_c=[d_0-a_0,\dots,d_{n-1}-a_{n-1}]\), we obtain 
\begin{IEEEeqnarray*}{rCl}
 K_c&=&
        \narrowbmatrix{ 3.7\!-\!2\i\!-\!0.33\j\!-\!0.67\k& 4\!+\!0.33\i\!-\!1.3\j\!+\!0.67\k},
    \\
 K&=&  \narrowbmatrix{3.3\!+\!0.33\j\!+\!2.7\k &   -2\!+\!1.7\i\!+\!0.33\j\!-\!0.67\k},
\end{IEEEeqnarray*}
and
\begin{IEEEeqnarray*}{rCl}
 A_{\rm cl} &=& 
    \narrowbmatrix{
        -2.3\!-\!0.33\j\!-\!2.7\k & 2\!-\!0.67\i\!-\!0.33\j\!+\!0.67\k \\
        2.7\!+\!0.33\i\!+\!\j\!-\!3.3\k & \!-\!0.67\!+\!0.33\i\!-\!1.7\j\!+\!3\k 
    }.
\end{IEEEeqnarray*}
Its right spectrum is
\[
\sigma_\err(A_{\mathrm{cl}})
=\{[-1+\i],\,[-2+\i]\}
=\{[-1+\j],\,[-2+\k]\}.
\]
confirming that the controllable-companion-form coefficient update achieves the prescribed classes even with quaternionic coefficients.
\end{Ex}

\section{Quaternionic Ackermann Formula}
In this section we derive a quaternionic Ackermann formula. Unlike the companion-form coefficient update, it applies only when the chosen target polynomial \(a_d\) has real (central) coefficients. This still covers real classes and spherical classes arising from complex-conjugate pairs, while the proof proceeds via controllable companion relations rather than determinants or Cayley--Hamilton.
\begin{Rem}
Throughout this section, the desired polynomial \(a_d\) is assumed to belong to \(\mathbb{R}[\lambda]\). Since its coefficients are real and hence central, the left and right scalar evaluations coincide, and the same holds for the corresponding matrix evaluations. Accordingly, we write simply \(a_d(q)\) and \(a_d(A)\), without the subscripts \(\ell\) or \(r\).
This restriction is essential for the Ackermann formula below; for quaternionic coefficients, the formula need not hold.
\end{Rem}

\begin{The}[\emph{Quaternionic Ackermann formula}]
\label{ThAckermann}
Let $A\in\HHnn$ and $B\in\HHn$ be such that the pair $(A, B)$ is controllable with $\mathcal{C} = \left[ B,AB,\ldots,A^{n-1}B \right] \in\HHnn$. Let
\begin{equation}
a_d(\lambda )=d_0+d_1\lambda+\cdots+d_{n-1}\lambda^{n-1}+\lambda^n
\in\RR[\lambda],
  \label{Eq:desiredpol}  
\end{equation}
be a chosen monic real polynomial whose zeros encode the desired closed-loop similarity classes and let $e_n=[0,\dots,0,1]^{\mathsf T}\in\mathbb{R}^{n}$ be the last canonical basis vector. Then the state-feedback row vector
\begin{IEEEeqnarray}{rCl}
  K &=& e_n^\T \mathcal{C}^{-1} a_d(A),
    \label{Eq:qacker}   
\end{IEEEeqnarray}
with $u=-Kx$ yields the closed-loop matrix $A-BK$
whose right spectrum is the set of similarity classes determined by the zeros of \(a_d(\lambda)\):
\begin{IEEEeqnarray}{rCl}
    \sigma_\err(A_{\mathrm{cl}})&=&\{[\lambda]:a_d(\lambda)=0\}.
\end{IEEEeqnarray}
Equivalently, if \(\lambda_1,\dots,\lambda_m\) are representatives of the distinct similarity classes of zeros of \(a_d\), then
\[
\sigma_\err(A_{\mathrm{cl}})=\{[\lambda_1],\dots,[\lambda_m]\}.
\]
The algebraic multiplicities are encoded by the polynomial \(a_d\) itself, not by the set \(\sigma_\err(A_{\mathrm{cl}})\).
\end{The}
\begin{Rem}
The Ackermann step is valid only when \(a_d \in \mathbb{R}[\lambda]\), because real scalars lie in the center of \(\mathbb{H}\), ensuring that \(a_d(A)\) is unambiguous. In this case: closed-loop poles are real or occur as complex-conjugate pairs; left/right zeros coincide; and any nonreal root $\gamma\in\CC$ actually stands for the entire spherical class $[\gamma]\subset\mathbb H$ (choosing $\gamma,\bar\gamma$ assigns that class). 

For an \(n\)th-order single-input design, each real zero contributes one degree, while each nonreal conjugate pair contributes two. Counting algebraic multiplicity, these contributions sum to \(n\).
\end{Rem}

The following Lemma is important in the Ackermann formula proof.
\begin{Lem}[\emph{Polynomial Intertwining}]\label{Lem0swap}
Let $p(\lambda)\in\RR[\lambda]$ be a polynomial with real coefficients. If an invertible matrix $T$ satisfies
\begin{IEEEeqnarray}{rCl}
 A T &=& T\! A_c
 \label{Eq0polintertwin1}
 \end{IEEEeqnarray}
then
\begin{IEEEeqnarray}{rCl}
p\!\left(A\right)\!T &=& Tp\!\left(A_c \right).
\label{Eq0polintertwin2}   
\end{IEEEeqnarray}
\end{Lem}
\begin{proof}
From $ AT =T A_c$ it follows by induction that
\begin{IEEEeqnarray}{rCl}
    A^k T &=& T A_c^{\,k}\quad (k=0,1,\dots).
    \label{Eq0polintertwink} 
\end{IEEEeqnarray}
Write $p(\lambda)=\sum_{k=0}^n p_k \lambda^k$ with
$p_k\in\mathbb{R}.$ Then
\begin{IEEEeqnarray*}{rCl}
    p(A)T
&=& \sum_{k=0}^n p_k A^k T
= \sum_{k=0}^n p_k T\!  A_c^{\,k}
\\
&=& T \sum_{k=0}^n p_k A_c^{\,k}
= T p(A_c),
\end{IEEEeqnarray*}
where we used that real coefficients lie in the center $Z(\mathbb H)=\mathbb R$ and commute with 
$T$ and with powers of $A_c.$
\end{proof}
\begin{Rem}
The lemma generally fails if $p\notin\mathbb{R}[\lambda]$. Although one may of course define
$p(A_c) = p(T^{-1}AT)$
by direct substitution, the intertwining identity
\begin{IEEEeqnarray}{rCl}
  p(A)\,T &\stackrel{?}{=} & T\,p(A_c).
\label{Eq:polNOTintertwin}   
\end{IEEEeqnarray}
need not hold once any coefficient of 
$p$ is a nonreal quaternion.
Indeed, take $A=\i$ and $T = \k$ so that $A_c=T^{-1}AT=-\i,$
and $p(\lambda)=\lambda+\i$. We have
$p(A)T=(\i+\i)k=2\i\k=-2\j$ but $Tp(A_c)=k(-\i+\i)=0$,
hence $p(A) T\neq T p(A_c)$.
This illustrates why the Ackermann proof below requires $p\in\mathbb{R}[\lambda]$: real coefficients ensure unambiguous evaluation and permit the coefficient–matrix “swap” used in the intertwining step.
\end{Rem}
\begin{proof}[Proof of Theorem~\ref{ThAckermann}]
Classical proofs over $\RR$ or $\CC$ lean on determinants and adjoints \cite{Kailath1980,Chen1999}, characteristic polynomials \cite{Chen1999,Franklin2019}, or Cayley--Hamilton; see also Ackermann’s original paper \cite{Ackermann1972}. These arguments do not transfer verbatim to $\HH$ because of noncommutativity and left/right evaluation. We therefore proceed differently.

\noindent\emph{Step 1: Controllable form and coefficient matching.}
Let an invertible similarity $T$ put $(A,B)$ into lower controllable companion form:
\begin{IEEEeqnarray*}{rCl}
A_c&=&T^{-1}AT,\qquad B_c=T^{-1}B.
\end{IEEEeqnarray*}
Write the open-loop companion polynomial as
\begin{IEEEeqnarray*}{rCl}
a(\lambda)=a_0+a_1\lambda+\cdots+a_{n-1}\lambda^{n-1}+\lambda^n .
\end{IEEEeqnarray*}
Given the desired real monic closed-loop polynomial
\begin{IEEEeqnarray}{rCl}
a_d(\lambda)=d_0+d_1\lambda+\cdots+d_{n-1}\lambda^{n-1}+\lambda^n,
\end{IEEEeqnarray}
the controllable-form feedback row
\begin{IEEEeqnarray}{rCl}
K_c&=&
\begin{bmatrix}
d_0-a_0&\ldots&d_{n-1}-a_{n-1}
\end{bmatrix}
\label{Eq0kcda2}
\end{IEEEeqnarray}
assigns, by Theorem~\ref{ThAcBcroots}, the right spectrum of \(A_c-B_cK_c\) to the set of similarity classes determined by the right roots of \(a_d\). 
The multiplicities are encoded in the factorization of \(a_d\), not in the set-valued spectrum \(\sigma_\err(\cdot)\).
The feedback in the original coordinates is
\begin{IEEEeqnarray*}{rCl}
K&=&K_cT^{-1}.
\end{IEEEeqnarray*}

\noindent\emph{Step 2: Controllable-form Ackermann equals coefficient matching.}
Let $\mathcal{C}_c=\left[B_c,A_cB_c,\ldots,A_c^{n-1}B_c\right]$ be the controllability matrix of the controllable form \eqref{Eq0AcBcI}. The controllable-form Ackermann expression
\begin{IEEEeqnarray}{rCl}
K_c&=&e_n^\T \mathcal{C}_c^{-1} a_d(A_c)
\label{Eq0qackercon}
\end{IEEEeqnarray}
produces exactly the same $K_c$ as above. Indeed, applying Theorem~\ref{ThCompanion} to the already-companion pair $(A_c,B_c)$ gives $T=I$, hence
\begin{IEEEeqnarray*}{rCl}
e_n^\top\mathcal{C}_c^{-1}=e_1^\top.
\end{IEEEeqnarray*}
Therefore
\begin{IEEEeqnarray*}{rCl}
K_c&=&e_1^\T a_d(A_c).
\end{IEEEeqnarray*}
Furthermore, the shape of $A_c$ yields
\begin{IEEEeqnarray*}{rCl}
e_1^\T A_c^k&=&e_{k+1}^\top,\qquad k=0,\dots,n-1,\\
e_1^\T A_c^n&=&e_1^\T A_c^{n-1}A_c
=e_n^\top A_c
=\big[-a_0,\dots,-a_{n-1}\big].
\end{IEEEeqnarray*}
Since the coefficients $d_i$ are real and hence central,
\begin{IEEEeqnarray*}{rCl}
K_c&=&e_1^\top a_d(A_c)\\
&=&\big[d_0,\dots,d_{n-1}\big]+\big[-a_0,\dots,-a_{n-1}\big]\\
&=&
\begin{bmatrix}
d_0-a_0&\ldots&d_{n-1}-a_{n-1}
\end{bmatrix},
\end{IEEEeqnarray*}
which is exactly the coefficient-matching gain \eqref{Eq0kcda2}.

\noindent\noindent\emph{Step 3: Coordinate-free form and mapping back.}
Let $\mathcal{C}$ be the controllability matrix \eqref{Eq0CONmat} of $(A,B)$. From the similarity, for $k\ge 0$ we have
\begin{IEEEeqnarray*}{rCl}
B&=&TB_c,\qquad A^kT=TA_c^{\,k}.
\end{IEEEeqnarray*}
Hence
\begin{IEEEeqnarray}{c}
\mathcal{C}=T\mathcal{C}_c
\;\Longrightarrow\;
\mathcal{C}^{-1}=\mathcal{C}_c^{-1}T^{-1},
\qquad
T=\mathcal{C}\mathcal{C}_c^{-1}.
\label{Eq0CiCc}
\end{IEEEeqnarray}
Because $a_d\in\mathbb{R}[\lambda]$, Lemma~\ref{Lem0swap} gives
\begin{IEEEeqnarray*}{rCl}
T^{-1}a_d(A)&=&a_d(A_c)T^{-1}.
\end{IEEEeqnarray*}
Therefore
\begin{IEEEeqnarray}{rCl}\label{Eq0ackercf}
K&=&e_n^\T \mathcal{C}^{-1} a_d(A)
=e_n^\top \mathcal{C}_c^{-1}T^{-1}a_d(A)\nonumber\\
&=&\big(e_n^\top \mathcal{C}_c^{-1}a_d(A_c)\big)T^{-1}
=K_cT^{-1},
\end{IEEEeqnarray}
since $K_c=e_n^\T \mathcal{C}_c^{-1}a_d(A_c)$.

Thus, the basis-invariant formula \eqref{Eq:qacker} produces exactly the same feedback as the controllable-form construction, completing the proof.
\end{proof}
\begin{Ex}[\emph{Example \ref{ExSSCompan} continued}] \label{ExAcker}
Continuing the previous example, we already have the controllability inverse $\mathcal{C}^{-1}$ from \eqref{Eq0CONi}; we now consider specific target polynomials. \\
\emph{(a) Real targets:}  $a_d(\lambda)=\lambda^2+3\lambda+2$ ($\lambda_1=-1,\lambda_2=-2$) gives 
\begin{IEEEeqnarray*}{rCl}
    a_d(A)=
    \begin{bmatrix}
      6+\k  &    4\i-\j \\
    -\i+4\j &   1+2\k
    \end{bmatrix}.
\end{IEEEeqnarray*} 
and the Ackermann formula \eqref{Eq:qacker} yields
\begin{IEEEeqnarray*}{rCl}
    K&=&
    \begin{bmatrix}
    2.5+\i+2.5\k &   -1.5+\i-1.5\k  
    \end{bmatrix},
\end{IEEEeqnarray*}
which equals \eqref{Eq0ExKreal}. The resulting closed-loop matrix
$A_{\rm cl}=A-BK$ is identical to \eqref{Eq0ExAclreal} and its right spectrum is $\sigma_\err(A_{\rm cl}) = \{[-1],[-2]\}$ as desired.\\
\emph{(b) Complex pair:} $a_d(\lambda)=\lambda^2+2\lambda+2$
($\lambda_{1,2}=-1\pm \i$) gives rise to
\begin{IEEEeqnarray*}{rCl}
    a_d(A)&=&
    \begin{bmatrix}
    5+\k &  3\i-\j \\
    -\i+3\j & 1+\k 
    \end{bmatrix},
\end{IEEEeqnarray*}
and
\begin{IEEEeqnarray*}{rCl}
 K=
    \begin{bmatrix}
    2+\i+2\k & -1+\i-\k  
    \end{bmatrix}     
\end{IEEEeqnarray*}
which equals \eqref{Eq0ExKcomplex}. The closed-loop matrix
\(A_{\rm cl}\) is identical to \eqref{Eq0ExAclcomplex} and
\(\sigma_\err(A_{\mathrm{cl}})=\{[-1+\i]\}\). \\
\emph{(c) Quaternionic coefficients, for which the formula fails:}
All displayed numerical values in this example are rounded independently from full-precision computations; therefore, the printed intermediate quantities are not intended to reproduce one another exactly.
The Ackermann formula cannot assign quaternionic target classes represented, for example, by \(\lambda_1=-1+\j\) and \(\lambda_2=-2+\k\). If one nevertheless applies~\eqref{Eq:qacker}, incorrectly interpreting
\[
a_d(A)\quad\text{as}\quad a_d(A)_{\err}=d_0I+d_1A+A^2,
\]
then Ackermann returns
\begin{IEEEeqnarray*}{rCl}
    K&=&
    \narrowbmatrix{
   3.5\!+\!1.3\i\!-\!0.5\j\!+\!3\k& -1.5\!+\!2\i\!-\!0.5\j\!-\!1.7\k
    }.
\end{IEEEeqnarray*}
However, this feedback results in
\begin{IEEEeqnarray*}{rCl}
A_{\mathrm{cl}}&=&
\narrowbmatrix{
-2.5\!-\!1.3\i\!+\!0.5\j\!-\!3\k & 1.5\!-\!\i\!+\!0.5\j\!+\!1.7\k\\[2pt]
3\!-\!0.5\i\!-\!0.33\j\!-\!3.5\k & -1.7\!-\!0.5\i\!-\!2\j\!+\!2.5\k}
\end{IEEEeqnarray*}
whose right spectrum is \(\{[-0.48+0.85\i],\,[-3.7+2.4\i]\}\), confirming that when \(a_d\notin\mathbb{R}[\lambda]\), the Ackermann construction does not enforce the prescribed classes.
\end{Ex}

\section{Conclusion}
This note developed a quaternionic framework for single-input state-feedback pole placement. We formalized controllable companion forms over $\HH$ and their companion polynomials satisfying \(a(A_c)_{\err}=0\), and presented two designs: coefficient matching in companion coordinates and a coordinate-free Ackermann gain for real (central) targets. The approach mirrors classical workflows while making quaternionic specifics explicit (similarity-class poles and order-sensitive evaluation).

Multi-input extensions should follow via block companion/Brunovsky forms with the same central-polynomial calculus. Future work includes observer duals, robustness, numerical conditioning, and output-feedback design.




\begin{thebibliography}{00}

\bibitem{Ackermann1972}
J.~Ackermann, ``Der Entwurf linearer Regelungssysteme im Zustandsraum,'' \emph{at -- Automatisierungstechnik}, vol.~20, no.~1--12, pp.~297--300, 1972, \DOI{10.1524/auto.1972.20.112.297}.

\bibitem{AntsaklisMichel06}
P.~J. Antsaklis and A.~N. Michel, \emph{Linear Systems}, 2nd~ed., Boston, MA, USA: Birkhäuser, 2006.

\bibitem{ChapmanMachen16}
A.~Chapman and C.~Machen, ``Standard polynomial equations over division algebras,'' \emph{Adv. Appl. Clifford Algebras}, vol.~27, pp.~1065--1072, 2017, \DOI{10.1007/s00006-016-0740-4}.

\bibitem{Chen1999}
C.-T. Chen, \emph{Linear System Theory and Design}, 3rd~ed. New York, NY, USA: Oxford Univ. Press, 1999.

\bibitem{daSilvaRouchon2008}
P.~S.~P. da~Silva and P.~Rouchon, ``Flatness-based control of a single qubit gate,'' \emph{IEEE Trans. Autom. Control}, vol.~53, no.~3, pp.~775--779, Apr. 2008, \DOI{10.1109/TAC.2008.917650}.

\bibitem{Emsley1992}
L.~Emsley and G.~Bodenhausen, ``Optimization of shaped selective pulses for NMR using a quaternion description of their overall propagators,'' \emph{J. Magn. Reson.}, vol.~97, no.~1, pp.~135--148, 1992, \DOI{10.1016/0022-2364(92)90242-Y}.

\bibitem{Franklin2019}
G.~F. Franklin, J.~D. Powell, and A.~Emami-Naeini, \emph{Feedback Control of Dynamic Systems}, 8th~ed. Boston, MA, USA: Pearson, 2019.

\bibitem{GordonMotzkin65}
B.~Gordon and T.~S. Motzkin, ``On the zeros of polynomials over division rings,'' \emph{Trans. Amer. Math. Soc.}, vol.~116, pp.~218--226, 1965, \DOI{10.1090/S0002-9947-1965-0195853-2}.

\bibitem{JiangLiuKouWang20}
B.~X. Jiang, Y.~Liu, K.~I. Kou, and Z.~Wang, ``Controllability and observability of linear quaternion-valued systems,'' \emph{Acta Math. Sin. (Engl. Ser.)}, vol.~36, no.~11, pp.~1299--1314, Nov. 2020, \DOI{10.1007/s10114-020-8167-1}.

\bibitem{Kailath1980}
T.~Kailath, \emph{Linear Systems}. Englewood Cliffs, NJ, USA: Prentice Hall, 1980.

\bibitem{SPMag2023}
N.~Le~Bihan \emph{et al.}, ``Quaternions in signal and image processing: A comprehensive and objective overview,'' \emph{IEEE Signal Processing Magazine}, vol.~40, no.~6, pp.~26--40, Sept. 2023, \DOI{10.1109/MSP.2023.3278071}.

\bibitem{Meglicki2008}
Z.~Meglicki, ``Quaternions,'' in \emph{Quantum Computing Without Magic: Device}. Cambridge, MA, USA: MIT Press, 2008.

\bibitem{Miron2006}
S.~Miron, N.~Le~Bihan, and J.~I. Mars, ``Quaternion-MUSIC for vector-sensor array processing,'' \emph{IEEE Trans. Signal Process.}, vol.~54, no.~4, pp.~1218--1229, Apr. 2006, \DOI{10.1109/TSP.2006.870630}.

\bibitem{Opfer09}
G.~Opfer, ``Polynomials and Vandermonde matrices over the field of quaternions,'' \emph{ETNA}, vol.~36, pp.~9--16, 2009.

\bibitem{Parcollet2018}
T.~Parcollet \emph{et al.}, ``Quaternion convolutional neural networks for end-to-end automatic speech recognition,'' in \emph{Proc. Interspeech},  pp.~22--26, 2018, \DOI{10.48550/arXiv.1806.07789}.

\bibitem{PereiraVettori2006}
R.~Pereira and P.~Vettori, ``Stability of quaternionic linear systems,'' \emph{IEEE Trans. Autom. Control}, vol.~51, no.~3, pp.~518--523, Mar. 2006, \DOI{10.1109/TAC.2005.864202}.

\bibitem{PogoruiShapiro2004}
A.~Pogorui and M.~Shapiro, ``On the Structure of the Set of Zeros of Quaternionic Polynomials,'' \emph{Complex Variables and Elliptic Equations}, vol.~49, no.~6, pp.~379--389, 2004, \DOI{10.1080/0278107042000220276}.

\bibitem{Rodman2014}
L.~Rodman, \emph{Topics in Quaternion Linear Algebra}. Princeton, NJ, USA: Princeton Univ. Press, 2014. 

\bibitem{Sangwine1996}
S.~J. Sangwine, ``Fourier transforms of colour images using quaternion (hypercomplex) numbers,'' \emph{Electronics Letters}, vol.~32, no.~21, pp.~1979--1980, Oct. 1996, \DOI{10.1049/el:19961331}.

\bibitem{Takahashi2017}
K.~Takahashi, A.~Isaka, T.~Fudaba, and M.~Hashimoto, ``Remarks on quaternion neural network-based controller trained by feedback error learning,'' in \emph{Proc. 2017 IEEE/SICE Int. Symp. on System Integration (SII)}, 2017, pp.~875--880, \DOI{10.1109/SII.2017.8279333}.

\bibitem{WhartonKoch15} %
K.~B. Wharton and D.~Koch, ``Unit quaternions and the Bloch sphere,'' \emph{J. Phys. A: Math. Theor.}, vol.~48, no.~23, p.~235302, May 2015, \DOI{10.1088/1751-8113/48/23/235302}.

\end{thebibliography}
\end{document}